\documentclass[11pt, letterpaper]{article}
\usepackage[utf8]{inputenc}
\usepackage[margin=1in]{geometry}
\usepackage{amssymb}
\usepackage{amsthm}
\usepackage{amsmath}
\usepackage{mathtools}
\usepackage{bbm}
\usepackage{xcolor}
\usepackage{subcaption}
\usepackage{caption}
\usepackage{algorithm}
\usepackage{algpseudocode}

\usepackage[nodayofweek,level]{datetime}
\usepackage{graphicx}
\usepackage{enumitem}
\usepackage{multicol}
\usepackage{hyperref}
\usepackage[capitalise]{cleveref}

\usepackage{amssymb}
\usepackage{amsthm}
\usepackage{amsmath}
\usepackage{mathtools}
\usepackage{bbm}
\usepackage{wasysym}

\theoremstyle{plain}
\newtheorem{theorem}{Theorem}
\newtheorem{lemma}{Lemma}

\newtheorem{conjecture}{Conjecture}

\theoremstyle{definition}
\newtheorem{definition}{Definition}

\newtheorem{sketch}{Sketch}

\newtheorem{example}{Example}

\theoremstyle{remark}
\newtheorem{remark}{Remark}

\DeclarePairedDelimiter{\ceil}{\lceil}{\rceil}

\DeclarePairedDelimiter{\norm}{\|}{\|}

\newcommand{\ind}[1]{\mathbbm{1}\left[{#1}\right]}

\newcommand{\pr}{\mathbb{P}}
\newcommand{\E}{\mathbb{E}}
\newcommand{\Exp}{\mathrm{Exp}}
\newcommand{\Uniform}{\mathrm{Uniform}}
\newcommand{\Poisson}{\mathrm{Poisson}}

\newcommand{\R}{\mathbb{R}}

\newcommand{\Z}{\mathbb{Z}}

\newcommand{\N}{\mathbb{N}}

\algtext*{EndIf}

\newcommand{\istrut}[2][0]{\rule[- #1 mm]{0mm}{#1 mm}\rule{0mm}{#2 mm}}

\newcommand{\Min}{\textsf{Min}}

\newcommand{\poly}{\text{poly}}

\newcommand{\erf}{\operatorname{erf}}
\newcommand{\Levy}{L\'{e}vy}
\newcommand{\Saglam}{Sa\v{g}lam}
\newcommand{\bydef}{\stackrel{\text{def}}{=}}
\newcommand{\Update}{\textsf{Update}}
\newcommand{\ParetoSampler}{\textsf{ParetoSampler}}
\newcommand{\Pareto}{\textsf{Pareto}}
\newcommand{\Sample}{\textsf{Sample}}
\newcommand{\argmin}{\operatorname{argmin}}
\newcommand{\Sampler}{\textsf{Sampler}}
\newcommand{\SamplerWOR}{\textsf{Sampler-WOR}}
\newcommand{\ParetoSamplerWOR}{\textsf{ParetoSampler-WOR}}
\newcommand{\KMIN}{k\textsf{-Min}}
\newcommand{\EdgeSampler}{\textsf{Edge-Sampler}}

\newcommand{\ignore}[1]{}

\begin{document}

\title{Universal Perfect Samplers for Incremental Streams\thanks{This work was supported by NSF Grant CCF-2221980.}}

 \author{Seth Pettie\\
 University of Michigan
 \and
 Dingyu Wang\\
 University of Michigan}

\date{}

\maketitle
\begin{abstract}
    If $G : \R_+ \to \R_+$, the \emph{$G$-moment} of a vector $\mathbf{x}\in\R_+^n$ is 
    $G(\mathbf{x}) \bydef \sum_{v\in[n]} G(\mathbf{x}(v))$ and
    the \emph{$G$-sampling} problem is to select an index $v_*\in [n]$
    according to its contribution to the $G$-moment, i.e., 
    such that $\pr(v_*=v) = G(\mathbf{x}(v))/G(\mathbf{x})$.  \emph{Approximate} 
    $G$-samplers may introduce multiplicative and/or additive errors
    to this probability, and some have a non-trivial probability of failure.

    In this paper we focus on the \emph{exact} $G$-sampling problem, 
    where $G$ is selected from the following class of functions.
    \[
    \mathcal{G}=\left\{G(z)=c\ind{z>0}+\gamma_0z+\int_0^\infty (1-e^{-rz})\,\nu(dr) \;\middle|\; c,\gamma_0\geq 0,\nu\text{ is non-negative}\right\}.
    \]
    The class $\mathcal{G}$ is perhaps more natural than it looks.  It captures all Laplace exponents of non-negative, one-dimensional \Levy{} processes,
    and includes several well studied classes such as 
    $p$th moments $G(z)=z^p$, $p\in[0,1]$,
    logarithms $G(z)=\log(1+z)$,
    Cohen and Geri's~\cite{CohenG19} \emph{soft concave sublinear} functions, which are used to approximate \emph{concave sublinear} functions, including \emph{cap statistics}.

\medskip 

    In this paper we develop $G$-samplers for a vector $\mathbf{x} \in \R_+^n$ 
    that is presented as an incremental stream of positive updates.  In particular:
    \begin{itemize}
    \item For any $G\in\mathcal{G}$, we give a very simple $G$-sampler that uses 2 words of memory and stores at all times a $v_*\in [n]$, such that $\pr(v_*=v)$ is  \emph{exactly} $G(\mathbf{x}(v))/G(\mathbf{x})$.

    \item We give a ``universal'' $\mathcal{G}$-sampler that 
    uses $O(\log n)$ words of memory w.h.p., and given any $G\in \mathcal{G}$
    at query time, produces an exact $G$-sample.
    \end{itemize}

    With an overhead of a factor of $k$, both samplers can be used to 
    $G$-sample a sequence of $k$ indices with or without replacement. 

\medskip 

    Our sampling framework is simple and versatile, 
    and can easily be generalized to sampling from more complex objects like graphs and hypergraphs.  
\end{abstract}

\section{Introduction}\label{sect:introduction}

We consider a vector $\mathbf{x}\in\R_+^n$, initially zero, 
that is subject to a stream of incremental (positive) updates:\footnote{$\R_+$ is the set of non-negative reals.}
\begin{itemize}
\item[] $\Update(v,\Delta)$ : Set $\mathbf{x}(v)\gets \mathbf{x}(v) + \Delta$, 
where $v\in [n], \Delta>0$.
\end{itemize}

The \emph{$G$-moment} of $\mathbf{x}$ is $G(\mathbf{x}) = \sum_{v\in[n]} G(\mathbf{x}(v))$ and a \emph{$G$-sampler} 
is a data structure that returns an index $v_*$ proportional 
to its contribution to the $G$-moment.

\begin{definition}[Approximate/Perfect/Truly Perfect 
$G$-samplers \cite{JayaramW21,JayaramWZ22,Jayaram21-PhDthesis}]\label{def:sampler}
Let $G:\R_+\to \R_+$ be a function.
An \emph{approximate $G$-sampler} with parameters $(\epsilon,\eta,\delta)$ is a sketch 
of $\mathbf{x}$ that can produce an 
index $v_*\in[n] \cup \{\perp\}$ 
such that $\pr(v_*=\,\perp)\leq \delta$ 
($v_*=\,\perp$ is failure) and
\[
\pr(v_* = v \mid v_* \neq\, \perp) \in (1 \pm \epsilon)G(\mathbf{x}(v))/G(\mathbf{x}) \pm \eta.
\]
If $(\epsilon,\eta) = (0,1/\poly(n))$ we say the sampler
is \emph{perfect} and if $(\epsilon,\eta)=(0,0)$ 
it is \emph{truly perfect}.
\end{definition}

In this paper we work in the \emph{random oracle} model
and assume we have access to a uniformly random hash function 
$H : [n]\to [0,1]$.

\subsection{Prior Work}

Much of the prior work on this problem 
considered $L_p$-samplers, $p\in [0,2]$,
in the more general turnstile model, 
i.e., $\mathbf{x}$ is subject to positive 
and negative updates and $G(z)=|z|^p$.
We survey this line of work, 
then review $G$-samplers in incremental streams.

\paragraph{$L_p$-Sampling from Turnstile Streams.}
Monemizadeh and Woodruff~\cite{MonemizadehW10}
introduced the first $\poly(\epsilon^{-1},\log n)$-space $L_p$-sampler.
(Unless stated otherwise, $\eta=1/\poly(n)$.)
These bounds were improved by Andoni, Krauthgamer,
and Onak~\cite{AndoniKO11} and then 
Jowhari, \Saglam, and Tardos~\cite{JowhariST11},
who established an upper bound of
$O(\epsilon^{-\max\{1,p\}}\log^2 n\log \delta^{-1})$ bits,
and an $\Omega(\log^2 n)$-bit lower bound 
whenever $\epsilon,\delta<1$.
Jayaram and Woodruff~\cite{JayaramW21} proved that
efficient $L_p$-samplers need not be approximate,
and that there are \emph{perfect} $L_p$ samplers 
occupying $O(\log^2 n\log\delta^{-1})$ bits when $p\in[0,2)$ 
and $O(\log^3 n\log\delta^{-1})$ bits when $p=2$.  
By concatenating $k$ independent $L_p$-samplers
one gets, in expectation, 
$(1-\delta)k$ independent samples \emph{with replacement}.
Cohen, Pagh, and Woodruff~\cite{CohenPW20} gave an
$O(k\poly(\log n))$-bit sketch that perfectly 
samples $k$ indices \emph{without replacement}.
Jayaram, Woodruff, and Zhou~\cite{JayaramWZ22} studied the distinction between \emph{perfect}
and \emph{truly perfect} samplers, proving that in turnstile streams,
perfect samplers require $\Omega(\min\{n,\log\eta^{-1}\})$-bits,
i.e., truly perfect sampling with non-trivial 
space is impossible.

\paragraph{$G$-Sampling from Incremental Streams.}
In incremental streams it is straightforward to 
sample according to the $F_0$ or $F_1$ frequency moments
(i.e., $G$-sampling with $G(z)=\ind{z>0}$ and $G(z)=z$, resp.\footnote{$\ind{\mathcal{E}}\in\{0,1\}$ is the indicator variable for the event/predicate $\mathcal{E}$.})
with a \Min-sketch~\cite{Cohen97} 
or reservoir sampling~\cite{vitter1985random}, respectively.
Jayaram, Woodruff, and Zhou~\cite{JayaramWZ22} gave 
truly perfect $G$-samplers for any monotonically increasing
$G:\N\to\R_+$ with $G(0)=0$, though the space used is
$\Omega(\frac{\|\mathbf{x}\|_1}{G(\mathbf{x})}\log\delta^{-1})$, 
which in many situations is $\Omega(\poly(n)\log\delta^{-1})$.\footnote{For example, take $G(z)=\sqrt{z}$ and $\mathbf{x}(1)=\cdots=\mathbf{x}(n)=n$,
then $\frac{\|\mathbf{x}\|_1}{G(\mathbf{x})}=n^2/n^{3/2}=\sqrt{n}$.}

Cohen and Geri~\cite{CohenG19} 
were interested in $G$-samplers for the class
of concave sublinear functions (\textsf{CSF}), 
which are those that can be expressed as
\[
\textsf{CSF} = \left\{G(z) = \int_0^\infty \min\{1,zt\}\nu(dt) \;\middle|\; \mbox{non-negative $\nu$}\right\}.
\]
This class can be approximated up to a constant factor
by the \emph{soft concave sublinear} functions (\textsf{SoftCSF}), 
namely those of the form
\[
\textsf{SoftCSF} = \left\{G(z) = \int_0^\infty (1-e^{-zt})\nu(dt) 
\;\middle|\; \mbox{non-negative $\nu$}\right\}.
\]
Cohen and Geri~\cite{CohenG19} developed 
\emph{approximate} $G$-samplers for $G\in \textsf{SoftCSF}$, 
a class that includes $F_p$ moments ($G(z)=z^p$),
for $p\in (0,1)$, and $G(z)=\log(1+z)$.
In their scheme there is a linear 
tradeoff between accuracy and update time.
Refer to Cohen~\cite{Cohen23-sampling-survey}
for a comprehensive survey of 
sampling from data streams and applications.

\subsection{New Results}

In this paper we build \emph{truly perfect} $G$-samplers 
with $(\epsilon,\eta,\delta)=(0,0,0)$ for any $G\in\mathcal{G}$.
\[
    \mathcal{G}=\left\{G(z)=c\ind{z>0}+\gamma_0z+\int_0^\infty (1-e^{-rz})\,\nu(dr)\right\},
\]
such that $c,\gamma_0\geq 0,$ $\nu$ 
is non-negative, and $\int_0^\infty \min\{t,1\}\nu(dt)<\infty$.

The class $\mathcal{G}$ is essentially the same as \textsf{SoftCSF},
but it is, in a sense, the ``right'' definition. 
According to the \Levy-Khintchine representation of \Levy{} processes, there is a bijection between the functions of $\mathcal{G}$ and the Laplace exponents of non-negative, one-dimensional \Levy{} processes, aka \emph{subordinators},
where the parameters $c,\gamma_0,\nu$ are referred to
as the \emph{killing rate}, the \emph{drift}, and the \emph{\Levy{} measure}, respectively.  
(\Levy{} processes and the \Levy-Khintchine representation are reviewed in~\cref{sect:Levy}.)

The connection between $\mathcal{G}$ and non-negative 
\Levy{} processes allows us to build simple, truly perfect samplers.
Given a $G\in\mathcal{G}$, let
$(X_t)_{t\ge 0}$, $X_t\in \R_+$,
be the corresponding \Levy{} process.
We define the
\emph{\Levy-induced level function} 
$\ell_G : \R_+ \times [0,1] \to \R_+$ to be
\[
\ell_G(a,b) = \inf\{t \;\mid\; \pr(X_t \geq a) \geq b\}.
\]

\alglanguage{pseudocode}
\begin{algorithm}[H]
\caption{Generic Perfect $G$-{\Sampler}}\label{alg:G-sampler}
{\bfseries Specifications:} The only state is a pair $(v_*,h_*)\in [n]\cup\{\perp\}\times \R_+\cup\{\infty\}$.  Initially $(v_*,h_*)=(\perp,\infty)$ and (implicitly) $\mathbf{x}=0^n$.  After processing a stream of vector updates $\{(v_i,\Delta_i)\}_i$, $(v_i,\Delta_i)\in [n]\times \R_+$,
$\pr(v_*=v) = G(\mathbf{x}(v))/G(\mathbf{x})$.  $H : [n]\to[0,1]$ is a hash function.\vspace{.2cm}
\begin{algorithmic}[1]
\Procedure{\Update}{$v,\Delta$} \Comment{$\mathbf{x}(v) \gets \mathbf{x}(v)+\Delta$}
\State Generate fresh $Y\sim \Exp(1)$.
\State $h \gets \ell_G(Y/\Delta,H(v))$ \Comment{$\ell_G$ is level function of $G$}
\If{$h<h_*$}
    \State $(v_*,h_*) \gets (v,h)$
\EndIf\EndProcedure
\end{algorithmic}
\end{algorithm}

The generic $G$-\Sampler{} (\cref{alg:G-sampler})
uses $\ell_G$ to sample an index 
$v$ proportional to $G(\mathbf{x}(v))$
with just \emph{2 words}\footnote{We assume a word stores an 
index in $[n]$ or a value in $\R_+$.  
See \cref{rem:precision} in \cref{sec:proofs-main-lemma-theorems} 
for a discussion of bounded-precision implementations.} of memory.
\begin{theorem}[$G$-Sampler]\label{thm:generic-G-sampler}
Fix any $G\in\mathcal{G}$. 
The generic $G$-\Sampler{} stores a pair 
$(v_*,h_*)\in[n]\times \R_+$ such that at all times,
$\pr(v_*=v) = G(\mathbf{x}(v))/G(\mathbf{x})$, i.e.,
it is a truly perfect $G$-sampler with zero probability of failure.
\end{theorem}

Since $\ell_G(a,b)$ is increasing in both arguments,
among all updates $\{(v_i,\Delta_i)\}$ to the 
generic $G$-\Sampler, the stored sample must 
correspond to a point on the (minimum) \emph{Pareto frontier} of 
$\{(Y_i/\Delta_i,H(v_i))\}$.  Thus, it is possible to 
produce a $G$-sample for any $G\in\mathcal{G}$ 
simply by storing the Pareto frontier.
(This observation was also used by Cohen~\cite{Cohen18} 
in her approximate samplers.)
The size of the Pareto frontier is a random variable
that is less than $\ln n+1$ in expectation and 
$O(\log n)$ with high probability.

\begin{theorem}\label{thm:ParetoSampler}
Suppose \ParetoSampler{} processes a stream of
$\poly(n)$ updates to $\mathbf{x}$.
The maximum space used is $O(\log n)$ 
words with probability $1-1/\poly(n)$.
At any time, given a $G\in\mathcal{G}$, 
it can produce a $v_*\in[n]$ such that
$\pr(v_*=v)=G(\mathbf{x}(v))/G(\mathbf{x})$.
\end{theorem}

\alglanguage{pseudocode}
\begin{algorithm}[H]
\caption{$\ParetoSampler$}\label{alg:Pareto-sampler}
{\bfseries Specifications:} 
The state is a set $S \subset \R_+ \times [0,1] \times [n]$, initially empty.  
The function $\Pareto(L)$ returns the (minimum) Pareto frontier of the tuples $L$ w.r.t.~their first two coordinates.\vspace{.2cm}
\begin{algorithmic}[1]
\Procedure{\Update}{$v,\Delta$} \Comment{$\mathbf{x}(v) \gets \mathbf{x}(v)+\Delta$}
\State Generate fresh $Y\sim \Exp(1)$.
\State $S \gets \Pareto(S \cup \{(Y/\Delta,H(v),v)\})$
\EndProcedure

\Procedure{\istrut{8}\Sample}{$G$}  \Comment{$G\in \mathcal{G}$}
\State Let $\ell_G : \R_+ \times [0,1] \to \R_+$ be the \emph{level function} of $G$
\State $(a_*,b_*,v_*) \gets \argmin_{(a,b,v)\in S} \{\ell_G(a,b)\}$
\State {\bfseries Return}$(v_*)$
\Comment{$v_*=v$ sampled with probability $G(\mathbf{x}(v))/G(\mathbf{x})$.}
\EndProcedure
\end{algorithmic}
\end{algorithm}
In \cref{sec:wor}, we show that both $G$-\Sampler{} and \ParetoSampler{} can be modified to sample \emph{without replacement} as well. 
One minor drawback of the generic $G$-\Sampler{} 
is that we have to compute the level function for $G$.
For specific functions $G$ of interest, 
we would like to have explicit, hardwired
expressions for the level function. 
An example of this is a new, simple 
$F_{1/2}$-\Sampler{} presented in \cref{alg:Fhalf-sampler}.
\emph{Why} Line 3 effects sampling according to the 
weight function $G(z)=z^{1/2}$ is explained 
in \cref{sec:deriving-level-functions}.
(Here $\erf^{-1}$ is the inverse Gauss error 
function, which is available as 
\texttt{scipy.special.erfinv} in Python.)

\alglanguage{pseudocode}
\begin{algorithm}[H]
\caption{$F_{1/2}$-{\Sampler}}\label{alg:Fhalf-sampler}
{\bfseries Specifications:} The state is $(v_*,h_*)\in [n]\cup\{\perp\}\times \R_+\cup\{\infty\}$, initially $(\perp,\infty)$. After processing a stream of updates,
        $\pr(v_*=v) = \sqrt{\mathbf{x}(v)}/\sum_{u\in [n]} \sqrt{\mathbf{x}(u)}$.\vspace{.2cm}
\begin{algorithmic}[1]
\Procedure{\Update}{$v,\Delta$} \Comment{$\mathbf{x}(v) \gets \mathbf{x}(v)+\Delta$}
\State Generate fresh $Y\sim \Exp(1)$.
\State $h \gets \sqrt{2Y/\Delta} \cdot \erf^{-1}(H(v))$ \Comment{$\erf:$ Gauss error function}
\If{$h<h_*$}
    \State $(v_*,h_*) \gets (v,h)$
\EndIf\EndProcedure
\end{algorithmic}
\end{algorithm}

\cref{lem:level} is the key that unlocks all of our results.  
It shows that for any $\lambda>0$, \Levy-induced level functions 
can be used to generate variables distributed according to $\Exp(G(\lambda))$, which are directly useful for truly perfect 
$G$-sampling and even $G$-moment estimation.
\begin{lemma}[level functions]\label{lem:level}
For any function $G\in \mathcal{G}$, there exists a 
(deterministic) 
function $\ell_G : (0,\infty) \times (0,1) \to \R_+$ satisfying:
\begin{description}
    \item[2D-monotonicity.] for any $a,a'\in \R_+$ and $b,b'\in[0,1]$, $a\leq a'$ and $b\leq b'$ implies $\ell_G(a,b)\leq \ell_G(a',b')$;
    \item[$G$-transformation.] if $Y\sim \Exp(\lambda)$ and 
    $U\sim \mathrm{Uniform}(0,1)$, 
    then $\ell_G(Y,U)\sim \Exp(G(\lambda))$.
\end{description}
\end{lemma}

\cref{lem:level} is a powerful tool. 
It can be used to build samplers
for objects more complex than vectors.  
To illustrate one example situation, suppose $H=([n],E)$ is a fixed graph 
(say a large grid) whose vector of vertex weights 
$\mathbf{x} \in \R_+^n$ are subject to a stream of incremental updates.
We would like to sample an edge $(u,v) \in E(H)$ 
proportional to its $G$-weight $G(\mathbf{x}(u),\mathbf{x}(v))$. 
We build $G$ using a \emph{stochastic sampling circuit}, 
whose constituent parts correspond to addition, scalar multiplication, 
and evaluating $\mathcal{G}$-functions.  
For example, in \cref{sec:sample_circuit} we show how to sample an edge
according to the edge weight
$G(a,b) = \log(1+\sqrt{a}+\sqrt{b}) + 2(1-e^{-a-b})$, 
which is constructed from $\mathcal{G}$-functions 
$\sqrt{x},1-e^{-x},$ and $\log(1+x)$.  
This approach can trivially be extended to 
$G$-sampling edges from hypergraphs, 
and to \emph{heterogenous} sampling, 
where we want each edge $(u,v)$ to be sampled proportional
to $G_{(u,v)}(\mathbf{x}(u),\mathbf{x}(v))$, where the $\mathcal{G}$-functions $\{G_{(u,v)}\}$ could all be different.






%




\subsection{Organization}

In \cref{sect:Levy} we review non-negative \Levy{} processes 
and the specialization of the \Levy-Khintchine representation 
theorem to non-negative processes.
In \cref{sec:proofs-main-lemma-theorems} we prove \cref{lem:level} 
and \cref{thm:generic-G-sampler,thm:ParetoSampler} on 
the correctness of the generic $G$-\Sampler{} and \ParetoSampler.
In \cref{sec:deriving-level-functions} 
we give explicit formulae for
the level functions of a variety of $\mathcal{G}$-functions,
including $G(z)=z^{1/2}$ (the $F_{1/2}$-\Sampler), 
the soft-cap sampler, 
$G(z)=1-e^{\tau z}$, 
and the log sampler, 
$G(z)=\log(1+z)$.
\cref{sec:sample_circuit} introduces 
stochastic sampling circuits, 
one application of which is 
$G$-edge-sampling from (hyper)graphs.
\cref{sec:conclusion} concludes with 
some remarks and open problems.
See \cref{sec:wor} for adaptations of our algorithms to sampling 
$k$ indices \emph{without} 
replacemenet.

\section{\Levy{} Processes and \Levy-Khintchine Representation}\label{sect:Levy}

\Levy{} processes are stochastic 
processes with independent, stationary increments.
In this paper we consider only one-dimensional, 
\emph{non-negative} \Levy{} processes.
This class excludes some natural processes
such as Wiener processes (Brownian motion).

\begin{definition}[non-negative L\'evy processes \cite{ken1999levy}]
    A random process $X=(X_t)_{t\geq 0}$ is a 
    non-negative \Levy{} process if it satisfies:
    \begin{description}
        \item[Non-negativity.] $X_t\in \R_+\cup\{\infty\}$ for all $t\in \R_+$ .\footnote{We \underline{do} allow the 
        random process 
        to take on the value $\infty$, which turns out to be meaningful and often useful for designing algorithms.}
        \item[Stationary Increments.] $X_{t+s}-X_{t}\sim X_s$ for all $ t,s\in \R_+$.\label{item:stationary}
        \item[Independent Increments.] For any $0\leq t_1 < t_2\ldots <t_k$, $X_{t_1},X_{t_2}-X_{t_1},\ldots, X_{t_k}-X_{t_{k-1}}$ are mutually independent.
        \label{item:independent}
        \item[Stochastic Continuity.] $X_0=0$ almost surely and $\lim_{t\searrow 0}\pr(X_t>\epsilon)=0$ for any $\epsilon>0$.
    \end{description}
\end{definition}

The bijection between $\mathcal{G}$ and 
non-negative \Levy{} processes is a 
consequence of the general \Levy-Khintchine 
representation theorem~\cite{ken1999levy}.
\begin{theorem}[\Levy-Khintchine representation for non-negative L\'evy processes. See Sato {\cite[Ch.~10]{ken1999levy}}]\label{thm:lk}
Any non-negative \Levy{} process $X=(X_t)_{t\geq 0}$ 
can be identified by a triple $(c,\gamma_0,\nu)$ where $c,\gamma_0\in\R_+$ and $\nu$ is a 
measure on $(0,\infty)$ such that
\begin{align}
    \int_{(0,\infty)}\min(r,1)\,\nu(d{r})<\infty.\label{eq:measure_cond}
\end{align}
The identification is through the Laplace transform. For any $t,z\in \R_+$
\begin{align}
    \E e^{-z X_t} &= \exp\left(-t \left(c\ind{z>0} + \gamma_0 z +\int_{(0,\infty)}(1-e^{-rz})\,\nu(dr)\right)\right). \label{eq:laplace}
\end{align}
Conversely, any triple $(c,\gamma_0,\nu)$ with $c,\gamma_0\in \R_+$ and $\nu$ satisfying (\ref{eq:measure_cond}) 
corresponds to a non-negative \Levy{} 
process $(X_t)$ satisfying (\ref{eq:laplace}).
\end{theorem}
The parameters $(c,\gamma_0,\nu)$ of a \Levy{} process $(X_t)_{t\ge 0}$ are called the 
\emph{killing rate}, the \emph{drift}, 
and the \emph{L\'evy measure}.
We associate with each $G\in\mathcal{G}$ 
a \emph{\Levy-induced level function}.
\begin{definition}[\Levy{} induced level function]\label{def:induced}
Let $G\in \mathcal{G}$ and $X=(X_t)_{t\geq 0}$ be the corresponding non-negative \Levy{} process, 
i.e., for any $t,z\in \R_+$,
\[
\E e^{-zX_t}=e^{-t G(z)}.
\]
The \emph{induced level function} 
$\ell_G:(0,\infty)\times (0,1)\to \R_+$ is,
for $a\in (0,\infty)$ and $b\in(0,1)$, defined to be
\[
    \ell_G(a,b) = \inf\{t \;\mid\; \pr(X_t \geq a)\geq  b\}.
\]
\end{definition}

\section{Proofs of \cref{lem:level} and \cref{thm:generic-G-sampler,thm:ParetoSampler}}\label{sec:proofs-main-lemma-theorems}

In this section we prove the key lemma stated in the introduction, \cref{lem:level}, as well as 
\cref{thm:generic-G-sampler,thm:ParetoSampler} 
concerning
the correctness of the generic $G$-\Sampler{}
and the $\mathcal{G}$-universal \ParetoSampler.

\begin{proof}[Proof of Lemma \ref{lem:level}]
Recall that 
$G\in\mathcal{G}$, 
$Y\sim \Exp(\lambda)$, and
$U\sim \mathrm{Uniform}(0,1)$.
We argue that $\ell_G(Y,U)$ (\cref{def:induced}) 
is monotonic in both arguments
and analyze its distribution.

By \Levy-Khintchine, $G$ has a corresponding non-negative 
\Levy{} process 
$X=(X_t)_{t\geq 0}$ for which $\E e^{-z X_t}=e^{-tG(z)}$. 
Note that since \Levy{} processes are memoryless, 
a non-negative \Levy{} process is also \emph{non-decreasing}.
Therefore $\pr(X_t\geq a)$ 
is increasing in $t$ and decreasing in $a$,
and as a consequence, 
$\ell_G(a,b)=\inf\{t \;\mid\; \pr(X_t\geq a)\geq b\}$ is 2D-monotonic. 
We now analyze the distribution of 
$\ell_G(Y,U)$.   For any $w>0$, we have 
    \begin{align*}
        \pr(\ell_G(Y,U)\geq w) &= \pr(\inf\{t:\pr(X_t\geq Y)\geq U\}\geq w) & \mbox{definition of $\ell_G$}
        \intertext{Note that by the definition of L\'evy process, $\pr(X_t\geq Y)$ is a continuous function in $t$ and therefore $\inf\{t:\pr(X_t\geq Y)\geq U\}\geq w$ is equivalent to $\pr(X_w\geq Y)\leq U$. Thus, this is equal to}
        &=\pr(\pr(X_w\geq Y)\leq U)\\
        &= 1-\pr(X_w\geq Y)  & U\sim\mathrm{Uniform}(0,1)\\
        &=\pr(X_w< Y)\\
        &= \E (\pr(X_w< Y) \mid X_w)\\
        &= \E (e^{-\lambda X_w} \mid X_w) & Y\sim \Exp(\lambda)\\
        &= \E e^{-\lambda X_w}\\
        &= e^{-w G(\lambda)} & \mbox{by \Levy-Khintchine (\cref{thm:lk})}.
    \end{align*}
Since the CDF of $\Exp(\lambda)$ is $1-e^{-\lambda x}$,
we conclude that $\ell_G(Y,U)\sim \Exp(G(\lambda))$.
\end{proof}

We can now prove the correctness of the generic $G$-\Sampler{} (\cref{thm:generic-G-sampler}, \cref{alg:G-sampler})
and the \ParetoSampler{} (\cref{thm:ParetoSampler},\cref{alg:Pareto-sampler}).

\begin{proof}[Proof of \cref{thm:generic-G-sampler} ($G$-\Sampler)]
Let $\mathbf{x} \in \R_+^n$ 
be the final vector after all updates,
and $(v_*,h_*)$ the final memory state of $G$-\Sampler.
We will prove that 
\begin{itemize}
    \item $h_*\sim \Exp(G(\mathbf{x}))$, and
    \item for any $v\in[n]$, $\pr(v_*=v)=G(\mathbf{x}(v))/G(\mathbf{x})$.
\end{itemize}
For $v\in [n]$, let $h_v$ be the smallest 
value produced by 
the $k$ updates $(v,\Delta_i)_{i\in [k]}$ to index $v$, 
so $\mathbf{x}(v) = \sum_{i=1}^{k} \Delta_i$.
Let $Y_1,\ldots,Y_{k}$ be the 
i.i.d.~$\Exp(1)$ random variables 
generated during those updates.
Then
\begin{align*}
    h_v &\sim \min \big\{\ell_G(Y_1/\Delta_1,H(v)),\ldots \ell_G(Y_{k}/\Delta_k,H(v))\big\}
    \intertext{and by the 2D-monotonicity property of \cref{lem:level}, this is equal to}
    &= \ell_G\big(\min\big\{Y_1/\Delta_1,\ldots,Y_{k}/\Delta_k\big\},H(v)\big). 
    \intertext{Note that $\min\{Y_1/\Delta_1,\ldots,Y_{k}/\Delta_k\}\sim \Exp(\sum_i \Delta_i) = \Exp(\mathbf{x}(v))$ 
    and $H(v)\sim \mathrm{Uniform}(0,1)$.
    By the $G$-transformation property of \cref{lem:level}, this is distributed as}
    &\sim \Exp(G(\mathbf{x}(v))).
\end{align*}
By properties of the exponential distribution, we have 
$h_* = \min_{v\in[n]} h_v \sim \min_{v\in [n]} \Exp(G(\mathbf{x}(v))) \sim \Exp(G(\mathbf{x}))$ and the probability that $v_*=v$ is sampled is exactly $G(\mathbf{x}(v))/G(\mathbf{x})$.
\end{proof}

\begin{remark}\label{rem:precision}
The proof of \cref{thm:generic-G-sampler} shows that $G$-\Sampler{} 
is truly perfect with zero probability of failure, according to \cref{def:sampler}, assuming the the value $h_*$ can be stored in a word. On a discrete computer, where the value $h_*$ can only be stored  discretely, there is no hope to have $h_*$ to distribute as $\Exp(G(\mathbf{x}))$ perfectly. Nevertheless, a truly perfect sample can still be returned on discrete computers, because one does not 
have to compute the exact $\ell_G(x,y)$ values 
but just correctly \emph{compare them}. 
We now discuss such a scheme. 
Let $\norm{\mathbf{x}}_\infty =\poly(n)$.
Supposing that we only stored $h_*$ to $O(\log n)$ bits of precision, 
we may not be able to correctly ascertain whether $h<h_*$
in Line 4 of $G$-\Sampler{} (\cref{alg:G-sampler}).  
This event occurs with probability\footnote{To see this, consider two independent random variables $Y_1\sim \Exp(\lambda_1)$ and $Y_2\sim \Exp(\lambda_2)$. By the properties of exponential random variables, conditioning on $Y_1<Y_2$, then $|Y_1-Y_2|\sim \Exp(\lambda_2)$; conditioning on $Y_1>Y_2$, then $|Y_1-Y_2|\sim \Exp(\lambda_1)$. This suggests $\pr(|Y_1-Y_2|\geq z)\geq e^{-\max(\lambda_1,\lambda_2) z}$. Thus if both $\lambda_1$ and $\lambda_2$ are $O(\poly(n))$, then with $|Y_1-Y_2|=\Omega(1/\poly(n))$ with probability $e^{-O(1/\poly (n))}\geq 1-O(1/\poly (n))$. Now let $\lambda_1 = \mathbf{x}({v_*'})$ and $\lambda_1 = \mathbf{x}({v_*})$, where $v_*'$ is the sample selected by the $O(\log n)$-bit sketch and $v_*$ is the sample selected by the infinite precision sketch (the first $O(\log n)$ bits are the same with the former one). If $v_*'\neq v_*$ then it implies $|Y_2-Y_1|<O(1/\poly(n))$, which happens only with probabilty $O(1/\poly(n))$.} $1/\poly(n)$, which induces an additive error
in the sampling probability, i.e.,
$(\epsilon,\eta,\delta)=(0,1/\poly(n),0)$.
We cannot regard this event 
as a \emph{failure} 
(with $(\epsilon,\eta,\delta)=(0,0,1/\poly(n))$)
because the sampling distribution, \emph{conditioned 
on non-failure}, would in general not be the same
as the truly perfect distribution.
There are ways to implement a truly perfect sampler
($(\epsilon,\eta,\delta)=(0,0,0)$) which affect the space bound.  Suppose update $(v,\Delta)$ is issued at (integer) 
time $k$.  We could store $(v,k)$ and generate $Y$ from $k$
via the random oracle.  Thus, in a stream of $m$ updates,
the space would be $O(\log(n + m))$ bits.
Another option is to generate more precise estimates
of $Y$ (and $H(v)$) on the fly.
Rather than store a tuple $(v,h)$, $h=\ell_G(Y/\Delta,H(v))$, 
we store $(v,R,\Delta)$, where $R\sim \mathrm{Uniform}(0,1)$
is dynamically generated to the precision necessary to execute
Line 4 of $G$-\Sampler{} (\cref{alg:G-sampler}).
Specifically, $Y = -\ln R \sim \Exp(1)$ is derived from $R$,
and if we cannot determine if $h<h_*$, 
where $h=\ell_G(Y/\Delta,H(v)), h_*=\ell_G(Y_*/\Delta_*,H(v_*))$, 
we append additional random bits to $R,R_*$ until 
the outcome of the comparison is certain.  
\end{remark}

\begin{proof}[Proof of \cref{thm:ParetoSampler} (\ParetoSampler)]
Let $T$ be the set of \emph{all}
tuples $(Y_i/\Delta_i,H(v_i),v_i)$ 
generated during updates, and $S$ be the minimum 
Pareto frontier of $T$ w.r.t.~the first two coordinates in each tuple.
Fix any query function $G\in\mathcal{G}$.
Imagine that we ran $G$-\Sampler{} (\cref{alg:G-sampler}) on the same update sequence
with the same randomness.  
By the 2D-monotonicity property of \cref{lem:level},
the output of $G$-\Sampler, 
$(v_i, h_* = \ell_G(Y_i/\Delta_i, H(v_i)))$
must correspond to a tuple $(Y_i/\Delta_i,H(v_i),v_i) \in S$
on the minimum Pareto frontier.
Thus, the {\sc Sample}$(G)$ function 
of \ParetoSampler{} would return
the same index $v_i$ as $G$-\Sampler.

We now analyze the space bound of \ParetoSampler.
Suppose that $h_v = \min\{Y_i/\Delta_i\}$, 
where the minimum is over all updates $(v,\Delta_i)$ 
to $\mathbf{x}(v)$.  (When $\mathbf{x}(v)=0$, $h_v=\infty$.)
We shall condition on \emph{arbitrary}
values $\{h_v\}$ and only consider the randomness introduced
by the hash function $H$, which effects a random permutation
on $\{h_v\}$.  Let $L=(h_{v_1},\ldots,h_{v_n})$ 
be the permutation of the $h$-values that is sorted 
in increasing order by $H(v_i)$.  
Then $|S|$ is exactly the number of distinct prefix-minima of $L$.  Define $X_i = \ind{h_{v_i} = \min\{h_{v_1},\ldots,h_{v_i}\}}$ to be the indicator that $h_{v_i}$ is a prefix-minimum and $v_i$ is included in a tuple of 
$S$.   Then
\[
\E(|S|) = \E\left(\sum_{i\in [n]} X_i\right) = \sum_{i\in [n]} 1/i 
= H_n < \ln n+1,
\]
where $H_n$ is the $n$th harmonic number. Note that $X_i$s are independent.
By Chernoff bounds, for any $c\geq 2$, 
$\pr(|S| > cH_n) < n^{-\Omega(c)}$.
Since there are only $\poly(n)$ updates, 
by a union bound,
$|S| = O(\log n)$ at all times, with probability
$1-1/\poly(n)$.
\end{proof}

\begin{remark}
The $O(\log n)$-word space bound holds 
even under some exponentially long update sequences.
Suppose all updates have magnitude at least 1, 
i.e., $\Delta_i\geq 1$.  
Then the same argument shows that
the expected number of times $h_v$ changes is 
at most $\ln\mathbf{x}(v)+1$ and by Chernoff bounds, 
the total number of times any of 
$\{h_v\}$ changes is 
$M = O(n\ln\|\mathbf{x}\|_\infty)$
with probability $1-\exp(-\Omega(M))$.
Thus, we can invoke a union bound
over $M$ states of the data structure and conclude $|S|=O(\log M)$ with probability $1-1/\poly(M)$.
This is $O(\log n)$ 
when $\|\mathbf{x}\|_\infty < \exp(\poly(n))$.
\end{remark}

\section{Deriving the Level Functions}\label{sec:deriving-level-functions}

The generic $G$-\Sampler{} and $\mathcal{G}$-universal \ParetoSampler{} refer to the \Levy-induced level function $\ell_G$.  
In this section we illustrate how to derive
expressions for $\ell_G$ in a variety of cases.

We begin by showing how $G$-\Sampler{} (\cref{alg:G-sampler})
``reconstructs'' the known $F_0$- and 
$F_1$-samplers, 
then consider a sample of non-trivial weight functions, 
$G(z) = z^{1/2}$ (used in the $F_{1/2}$-\Sampler, \cref{alg:Fhalf-sampler}), 
$G(z) = 1-e^{-\tau z}$ (corresponding to a Poisson process),
and $G(z) = \log(1+z)$ (corresponding to a Gamma process).

\begin{example}[$F_0$-sampler $\mapsto$  
$\textsf{Min}$ sketch~\cite{Cohen97}]\label{exa:f0_sampler}
The weight function for $F_0$-sampling is $G(z)=\ind{z>0}$. 
By L\'evy-Khintchine (\cref{thm:lk}), this function corresponds to a ``pure-killed process'' $(X_t)_{t\ge 0}$ 
which can be simulated as follows.
\begin{itemize}
    \item Sample a \emph{kill time} $Y\sim\Exp(1)$.
    \item Set $\displaystyle X_t = \left\{\begin{array}{ll}
    0 & \mbox{ if $t<Y$,}\\ 
    \infty & \mbox{ if $t\geq Y$.}\end{array}\right.$
\end{itemize}
The induced level function (\cref{def:induced}) is, for $a\in(0,\infty)$ and $b\in (0,1)$,
\begin{align*}
    \ell_G(a,b)&=\inf\{z \;\mid\; \pr(X_z\geq a)\geq b\}
    \intertext{By definition, $X_z\geq a>0$ if and only if $X$ has been killed by time $z$, i.e., $z\geq Y$.  Continuing,}
    &=\inf\{z \;\mid\; \pr(z\geq Y)\geq b\}\\
    &=\inf\{z \;\mid\; 1-e^{-z}\geq b\}\\
    &=-\log (1-b).
\end{align*}
Thus, by inserting this $\ell_G$ into $G$-\Sampler{} (\cref{alg:G-sampler}), it stores $(v_*,h_*)$ where $v_*$ has the 
smallest hash value and $h_*\sim \Exp(F_0)$,\footnote{Recall that for a frequency vector $\mathbf{x}\in\R_+^n$, $F_0=\sum_{v\in [n]}\ind{\mathbf{x}(v)>0}$ is the number of distinct elements present in the stream.} 
thereby essentially reproducing Cohen's~\cite{Cohen97} $\textsf{Min}$ sketch. 

\end{example}\begin{example}[$F_1$-sampler $\mapsto$ min-based reservoir sampling~\cite{vitter1985random}]\label{exa:f1_sampler}
The weight function for an $F_1$-sampler is $G(z)=z$. 
By \Levy-Khintchine, this corresponds 
to a deterministic drift process $(X_t)_{t\ge 0}$, 
where $X_t = t$. The induced level function is,
\begin{align*}
    \ell_G(a,b)&=\inf\{z \;\mid\; \pr(X_z\geq a)\geq b\}
    \intertext{and since $X_z=z$, $\pr(X_z\geq x)=\ind{z\geq x}$,}
    &=\inf\{z \;\mid\; \ind{z\geq a}\geq b\}\\
    &=a.  & \mbox{Note that $b\in(0,1)$}
\end{align*}
Thus the corresponding $G$-sampler does not use the hash function $H$,
and recreates reservoir sampling~\cite{vitter1985random} 
with the choice of replacement implemented 
by taking the minimum random value.
\end{example}

\medskip 

Next, we demonstrate a non-trivial application: 
the construction of the $F_{1/2}$-\Sampler{}
presented in \cref{alg:Fhalf-sampler}.

\begin{example}[$F_{1/2}$-sampler]\label{example:1/2-stable}
For $F_{1/2}$, the weight function is $G(z)=\sqrt{z}$, which corresponds to the non-negative $1/2$-stable process $(X_t)_{t\ge 0}$. 
The induced level function is
\begin{align*}
    \ell_G(a,b) &= \inf\{z \;\mid\; \pr(X_z\geq a)\geq b\}
    \intertext{and since $X$ is $1/2$-stable, we have $X_z\sim z^{2} X_{1}$.  Continuing,}
    &= \inf\{z \;\mid\; \pr(z^{2} X_{1}\geq a)\geq b\}.
    \intertext{It is known that the 1/2-stable $X_1$ distributes identically with $1/Z^2$ where $Z$ is a standard Gaussian~\cite[page~29]{ken1999levy}. 
    Thus, we have $\pr(X_1\geq r)=\pr\left(|Z|\leq \sqrt{1/r}\right)=\mathrm{erf}\left(\sqrt{\frac{1}{2r}}\right)$, 
    where $\mathrm{erf}(s)=\frac{2}{\sqrt{\pi}}\int_0^se^{-s^2}\,dx$ is the \emph{Gauss error function}.  As $\pr(z^2X_1 \geq a) = \pr(X_1 \geq a/z^2) = \mathrm{erf}(z/\sqrt{2a})$, this is equal to}
    &= \sqrt{2 a} \cdot \mathrm{erf}^{-1}(b).
\end{align*}
Plugging the expression 
$\ell_G(a,b) = \sqrt{2a}\cdot\erf^{-1}(b)$ 
into $G$-\Sampler, we arrive at the 
$F_{1/2}$-\Sampler{} of \cref{alg:Fhalf-sampler},
and thereby establish its correctness.
\end{example}

The $1/2$-stable distribution has a clean form, which yields a closed-form expression for the corresponding level function.  
Refer to Penson and G\'orska \cite{penson2010exact} for 
explicit formulae for one-sided $k/l$-stable distributions, 
where $k,l\in \Z_+$ and $k\leq l$, 
which can be used to write $F_{k/l}$-samplers 
with explicit level functions.

\medskip 

Previously, approximate ``soft cap''-samplers were used by Cohen to estimate cap-statistics \cite{Cohen18}. The weight function for a  soft cap sampler is parameterized by $\tau>0$, where $G_\tau(z)=1-e^{-\tau z}$. We now compute the level functions needed to build soft cap samplers with precisely correct sampling probabilities.
\begin{example}[``soft cap''-sampler]\label{example-softcap}
The weight function is $G_\tau(z)=1-e^{-\tau z}$. By L\'evy-Khintchine, $G_\tau$ corresponds to a unit-rate Poisson counting process $(X_t)_{t\ge 0}$ with jump size $\tau$, 
where $X_t/\tau \sim \Poisson(t)$.
By \cref{def:induced},
\begin{align*}
    \ell_{G_\tau}(a,b) &= \inf\{z \;\mid\; \pr(X_z\geq a)\geq b\}
    \intertext{Since $X_z/\tau\sim \Poisson(z)$, $\pr(X_z\geq a)=e^{- z}\sum_{j=\ceil{a\tau}}^{\infty} \frac{z^j}{j!}$.  Continuing,}
    &= \inf\left\{z \;\middle|\; e^{- z}\sum_{j=\ceil{a\tau}}^{\infty} \frac{z^j}{j!}\geq b\right\}.
\end{align*}
Thus $\ell_{G_{\tau}}(a,b)=w_{\ceil{a\tau},b}$ where $w_{k,b}$ is the unique solution to the equation $e^{-w}\sum_{j=k}^{\infty} \frac{w^j}{j!}= b$. Note that for any $k\in \Z_+$, the function $g(w)=e^{-w}\sum_{j=k}^{\infty} \frac{w^j}{j!}=\pr(\Poisson(w)\geq k)$ is increasing from 0 to $\infty$ as $w$ increases from $0$ to $\infty$, by a simple coupling argument. 
Therefore, $w_{k,b}$, as the solution of $g(w)=b$, 
can be computed with a binary search.
\end{example}

\begin{example}[log-sampler]\label{example:log-sampler}
Consider the weight function $G(z)=\log(1+z)$. By L\'evy-Khintchine, the corresponding process is a Gamma process $(X_t)_{t\ge 0}$, where $X_t\sim \mathrm{Gamma}(t,1)$. The PDF 
of $\mathrm{Gamma}(t,1)$ is $f(r)=r^{t-1}e^{-r}/\Gamma(t)$. By \cref{def:induced},
\begin{align*}
    \ell_{G_\tau}(a,b) &= \inf\{z \;\mid\; \pr(X_z\geq a)\geq b\}
    \intertext{ Since $X_z\sim \mathrm{Gamma}(z,1)$, $\pr(X_z\geq a)=\Gamma(z)^{-1}\int_a^\infty r^{z-1}e^{-r} \,dr$.  Continuing,}
    &= \inf\left\{z \;\middle|\;\Gamma(z)^{-1}\int_a^\infty r^{z-1}e^{-r} \,dr\geq b\right\}.
\end{align*}
Thus $\ell_{G}(a,b)=w'_{a,b}$ which is the unique solution of the equation $\Gamma(w)^{-1}\int_a^\infty r^{w-1}e^{-r} \,dr= b$. 
Once again the left hand side is monotonic in $w$ 
and therefore $w'_{a,b}$ can be found with a binary search.
\end{example}

As discussed in \cref{rem:precision}, there is no need to compute the level functions exactly, which is impossible to do so in the practical finite-precision model anyway. We only need to evaluate level functions to tell which index $v$ has the smallest $\ell_G(x,y)$ value. Such samples are still truly perfect, even though the final value $h_*$ is not a perfect $\Exp(G(\mathbf{x}))$ random variable. For a generic $G$, in practice one may pre-compute the level function of $G$ on a geometrically spaced lattice and cache it as a read-only table. Such a table can be shared and read simultaneously by an unbounded number of $G$-samplers for different applications and therefore the amortized space overhead is typically small.

\section{Stochastic Sampling Circuits}\label{sec:sample_circuit}

We demonstrate how sampling via level functions $\ell_G$ 
can be used in a more general context. 
Just as currents and voltages can 
represent signals/numbers, one may consider 
an exponential random variable $\sim \Exp(\lambda)$ 
as a \emph{signal} carrying information about its 
rate $\lambda$. 
These signals can be summed, scaled, and transformed as follows.
\begin{description}
    \item[Summation.] Given $Y_1\sim \Exp(\lambda_1),Y_2\sim \Exp(\lambda_2)$, $\min(Y_1,Y_2)\sim \Exp(\lambda_1+\lambda_2)$.
    \item[Scaling.] Fix a scalar $\alpha>0$. 
    Given $Y\sim \Exp(\lambda)$, $Y/\alpha \sim \Exp(\alpha\lambda)$.
    \item[$G$-transformation.] Given $Z\sim \Exp(\lambda)$ and $Y\sim \Uniform(0,1)$, $\ell_G(Z,Y)\sim \Exp(G(\lambda))$.
\end{description}

\medskip 

A \emph{stochastic sampling circuit} is an object that
uses \emph{summation}, \emph{scaling}, and 
\emph{$G$-transformation} gates to sample according
to functions with potentially many inputs.
Such a circuit is represented by a directed acyclic graph 
$(V,E)$ where $V$ is the set of \emph{gates} 
and $E\subset V\times V$ is a set of \emph{wires}. 
There are four types of gates.
\begin{itemize}
    \item An \emph{input-gate} $u$ receives a stream of 
    incremental updates. 
    Whenever it receives $\Delta>0$, 
    it generates, 
for \emph{each} outgoing edge $(u,v)$, a freshly sampled i.i.d.~$Y\sim \Exp(1)$ 
random variable and sends $Y/\Delta$ to $v$.
    \item A \emph{scalar-gate} $v$ is parameterized by a fixed $\alpha>0$, and has a unique predecessor $v'$ and successor $v''$, 
    $(v',v),(v,v'')\in E$.  
    Whenever $v$ receives a $y\in\R_+$ from $v'$ is sends
    $y/\alpha$ to $v''$.
    \item A \emph{$G$-gate} $v$ has a unique 
    successor $v'$, $(v,v')\in E$.  It is initialized with a random seed $U\sim\Uniform(0,1)$.  Whenever $v$ receives a number $y\in\R_+$ from a predecessor $u$, 
    it sends $\ell_G(y,U)$ to $v'$, 
    where $\ell_G$ is the \Levy-induced 
    level function of $G\in\mathcal{G}$. 
    \item An \emph{output-gate} stores a pair $(v_*,h_*)$, initialized as $(\perp,\infty)$. Whenever an output-gate receives a number $y\in \R_+$ from a predecessor 
    with id $v\in V$, 
    if $y<h_*$, then it sets $(v_*,h_*)\gets (v,y)$.
\end{itemize}

The restriction that $G$-gates and scalar-gates have 
only one successor guarantees that the numbers 
received by one gate from different wires are independent. 
The generic $G$-\Sampler{} (\cref{alg:G-sampler}) 
can be viewed as a \emph{flat} stochastic sampling circuit (\cref{fig:G_circuit}), where each element $v\in[n]$ has its own input-gate and $G$-gate.   We could just as easily assign each input gate $v$ to 
a $G_v$-gate, $G_v\in\mathcal{G}$, which would result in a \emph{heterogeneous} sampler, where $v$ is sampled with probability
$G_v(\mathbf{x}(v))/\sum_u G_u(\mathbf{x}(u))$.

\begin{figure}
    \centering
    \includegraphics{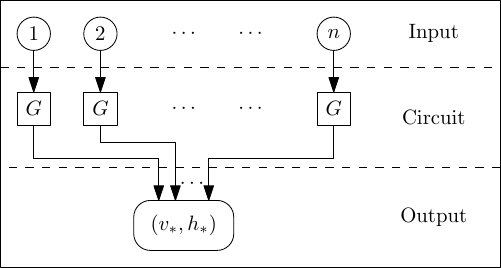}
    \caption{The flat stochastic sampling circuit corresponding to the generic $G$-\Sampler{} (\cref{alg:G-sampler}).}
    \label{fig:G_circuit}
\end{figure}

    \begin{figure}
    \centering
    \includegraphics[width=0.6\textwidth]{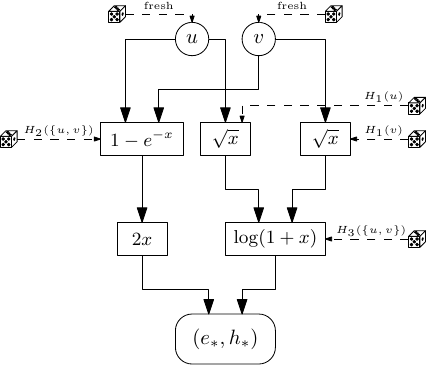}
    \caption{Circuit diagram for $G$-\EdgeSampler{} (\cref{alg:G-edge-sampler})
    with $G(a,b) = \log(1+\sqrt{a}+\sqrt{b}) + 2(1-e^{-(a+b)})$.
    Depicted are two of the input gates ($u$ and $v$), the unique output gate, and all gates related to the sampling of edge $\{u,v\}$.
    The dice symbol indicates the random seed $U \sim \mathrm{Uniform}(0,1)$ used by
    each $G$-gate, as well as the ``fresh'' exponential random variables $Y\sim \Exp(1)$ generated by the input gates for each update and each output wire.
     The ``$2x$''-gate is a deterministic 
     scalar gate with $\alpha=2$.\label{fig:edge_sample}}
    \end{figure}

\medskip 
We illustrate how stochastic sampling circuits can be used
to sample an edge from a graph with probability 
proportional to its weight.  Let $H=([n],E)$ 
be a fixed graph and $\mathbf{x} \in \R_+^n$ be 
a vector of vertex weights subject to incremental updates.  
For a fixed edge-weight function 
$G : \R_+^2 \to \R_+$,
where the weight of $(u,v)$ is 
$G(\mathbf{x}(u),\mathbf{x}(v))$, 
we would like to select
an edge $(u_*,v_*)$ with probability exactly
$G(\mathbf{x}(u_*),\mathbf{x}(v_*))/\sum_{(u,v)\in E} 
G(\mathbf{x}(u),\mathbf{x}(v))$.
Our running example is a symmetric weight function
that exhibits summation, scalar multiplication, and 
a variety of $\mathcal{G}$-functions.
\[
G(a,b) = \log(1+\sqrt{a}+\sqrt{b}) + 2(1-e^{-(a+b)}).
\]
The stochastic sampling circuit corresponding to $G$
is depicted in \cref{fig:edge_sample}.  
It uses the following level and hash functions.

\begin{itemize}
      \item $\ell_1$, level function for $G_1(x)=\sqrt{x}$. See \cref{example:1/2-stable}.
      \item $\ell_2$, level function for $G_2(x)=1-e^{-x}$. See \cref{example-softcap}.
      \item $\ell_3$, level function for $G_3(x)=\log(1+x)$. See \cref{example:log-sampler}.
      \item $H_1:V\to (0,1)$ uniformly at random.
      \item $H_2, H_3 : E\to (0,1)$ uniformly at random.
\end{itemize}

The implementation of this circuit as a 
$G$-\EdgeSampler{} is given in \cref{alg:G-edge-sampler}.
Note that since $G$ is symmetric, we can regard 
$H$ as an undirected graph and let \textsc{Update}$(v,\Delta)$
treat all edges incident to $v$ in the same way.
If $G$ were not symmetric, the code for $G$-\EdgeSampler{}
would have two {\bfseries for} loops, 
one for outgoing edges $(v,u)\in E$, 
and one for incoming edges $(u,v)\in E$.

\alglanguage{pseudocode}
\begin{algorithm}[H]
\caption{$G$-{\EdgeSampler} for $G(a,b) = \log(1+\sqrt{a}+\sqrt{b})+2(1-e^{-(a+b)})$}\label{alg:G-edge-sampler}
{\bfseries Specifications:} $H=([n],E)$ is a fixed graph. 
The state is $(e_*,h_*)\in {[n]\choose 2}\cup\{\perp\}\times \R_+\cup\{\infty\}$, initially $(\perp,\infty)$. 
After processing a stream of updates,
        $\pr(e_*=\{u_*,v_*\}) = G(\mathbf{x}(u_*),\mathbf{x}(v_*))/\sum_{\{u,v\}\in E}G(\mathbf{x}(u),\mathbf{x}(v))$.\vspace{.2cm}
\begin{algorithmic}[1]
\Procedure{\Update}{$v,\Delta$} \Comment{$\mathbf{x}(v) \gets \mathbf{x}(v)+\Delta$}
\For{each edge $\{u,v\}$ adjacent to $v$}
    \State Generate fresh, independent $Y_1,Y_2\sim \Exp(1)$
    \State $h\gets \min\{\ell_3(\ell_1(Y_1/\Delta,H_1(v)),H_3(\{u,v\})),\, \ell_2(Y_2/\Delta, H_2(\{u,v\}))/2\}$
    \If{$h < h_*$}
        \State $(e_*,h_*) \gets (\{u,v\}, h)$
    \EndIf
\EndFor
\EndProcedure
\end{algorithmic}
\end{algorithm}

\section{Conclusion}\label{sec:conclusion}

    In this paper we developed very simple sketches for 
    perfect $G$-sampling using $O(\log n)$ bits
    and universal perfect $\mathcal{G}$-sampling using $O(\log^2 n)$ bits.
    (See \cref{rem:precision} in \cref{sec:proofs-main-lemma-theorems} for a discussion of \emph{truly perfect} implementations.)
    They were made possible by the explicit connection between the class $\mathcal{G}$ and the Laplace exponents of non-negative \Levy{} processes via the \Levy-Khintchine representation theorem (\cref{thm:lk}).  
    To our knowledge this is the first explicit 
    use of the \Levy-Khintchine theorem in algorithm design, 
    though the class $\mathcal{G}$ was investigated 
    without using this connection, by Cohen~\cite{Cohen17,Cohen18,Cohen23-sampling-survey} 
    and Cohen and Geri~\cite{CohenG19}.
    A natural question is whether $\mathcal{G}$ captures all functions that have minimal-size $O(\log n)$-bit perfect samplers.

\begin{conjecture}
    Suppose $\mathbf{x} \in \R_+^n$ 
    is updated by an incremental stream.
    If there is an $O(\log n)$-bit 
    perfect $G$-sampler in the random oracle model
    (i.e., an index $v\in[n]$ is sampled
    with probability $G(\mathbf{x}(v))/G(\mathbf{x}) \pm 1/\poly(n)$), 
    then $G\in\mathcal{G}$.
\end{conjecture}

    Recall that $\mathcal{G}$ is in correspondence with 
    \emph{non-negative}, one-dimensional \Levy{} processes,
    which is just a small subset of all 
    \Levy{} processes.  It leaves out processes over $\R^d$,
    compound Poisson processes whose jump 
    distribution includes positive and negative jumps, 
    and $p$-stable processes for $p\in[1,2]$, among
    others.  Exploring the connection between general 
    \Levy{} processes, \Levy-Khintchine representation, 
    and data sketches is a promising direction for 
    future research.

\appendix

\section{Sampling Without Replacement}\label{sec:wor}

We can take $k$ independent copies of the $G$-\Sampler{} or \ParetoSampler{} sketches to sample $k$ indices from the $(G(\mathbf{x}(v))/G(\mathbf{x}))_{v\in[n]}$ distribution \emph{with} replacement.  A small change to these algorithms will 
sample $k$ indices \emph{without} replacement.  
See Cohen, Pagh, and Woodruff~\cite{CohenPW20}
for an extensive discussion of why WOR (without replacement)
samplers are often more desirable in practice.
The algorithm $(G,k)$-\SamplerWOR (\cref{alg:G-sampler-WOR})
samples $k$ (distinct) indices without replacement.

\alglanguage{pseudocode}
\begin{algorithm}[H]
\caption{$(G,k)$-{\SamplerWOR}}\label{alg:G-sampler-WOR}
{\bfseries Specifications:} The state is a set 
$S \subset [n]\times \R_+$, initially empty.
The function $\KMIN(L)$ takes a list $L\subset [n]\times \R_+$,
discards any $(v,h)\in L$ if there is a $(v,h')\in L$ with $h'<h$, then returns the $k$ elements with the smallest second coordinate.\vspace{.2cm}
\begin{algorithmic}[1]
\Procedure{\Update}{$v,\Delta$} \Comment{$\mathbf{x}(v) \gets \mathbf{x}(v)+\Delta$}
\State Generate fresh $Y\sim \Exp(1)$.
\State $h \gets \ell_G(Y/\Delta,H(v))$ \Comment{$\ell_G$ is level function of $G$}
\State $S \gets \KMIN(S\cup \{(v,h)\})$
\EndProcedure
\end{algorithmic}
\end{algorithm}

In a similar fashion, one can define a sketch $k$-\ParetoSamplerWOR{} 
analogous to $(G,k)$-\SamplerWOR,
that maintains the minimum $k$-Pareto frontier,
defined by discarding any tuple $(a,b,v)$ if there is 
another $(a',b,v)$ with $a'<a$, then retaining
only those tuples that are dominated by at most $k-1$ other tuples.

\begin{theorem}\label{thm:WOR}
    Consider a stream a $\poly(n)$ incremental updates to a vector $\mathbf{x}\in\R_+^n$.
    The $(G,k)$-\SamplerWOR{} occupies $2k$ words of memory,
    and can report an ordered tuple $(v_*^{1},\ldots,v_*^{k}) \in [n]^k$
    such that 
\begin{align}
\pr((v_*^{1},\ldots,v_*^{k}) = (v^{1},\ldots,v^k))
= \prod_{i=1}^{k} \frac{G(\mathbf{x}(v^i))}{G(\mathbf{x}) - \sum_{j=1}^{i-1} G(\mathbf{x}(v^j))}.\label{eqn:WOR}
\end{align}
    The $k$-\ParetoSampler{} occupies $O(k\log n)$ words w.h.p.~and for any $G\in\mathcal{G}$ at query time,
    can report a tuple $(v_*^{1},\ldots,v_*^{k}) \in [n]^k$ distributed according to \cref{eqn:WOR}.
\end{theorem}

\begin{proof}
    The proof of \cref{thm:generic-G-sampler} shows that
    $h_v \sim \Exp(G(\mathbf{x}(v)))$ and if 
    $v_*^1$ minimizes $h_{v_*^1}$, that
    $h_{v_*^1} \sim \Exp(G(\mathbf{x}))$. 
    It follows that $\pr(v_*^1=v)=G(\mathbf{x}(v))/G(\mathbf{x})$.
    By the memoryless property of the exponential distribution,
    for any $v\neq v_*^1,$ 
    $h_v - h_{v_*^1} \sim \Exp(G(\mathbf{x}(v)))$,
    hence $h_{v_*^2} \sim \Exp(G(\mathbf{x})-G(\mathbf{x}(v_*^1)))$ and 
    $\pr(v_*^2 = v \mid v_*^1, v\neq v_*^1) = G(\mathbf{x}(v))/(G(\mathbf{x})-G(\mathbf{x}(v_*^1)))$.  The distribution of $v_*^3,\ldots,v_*^k$ is analyzed in the same way.

    By the 2D-monotonicity property, 
    the $k$-Pareto frontier contains all the 
    points that would be returned by $(G,k)$-\SamplerWOR,
    hence the output distribution of $k$-\ParetoSamplerWOR{}
    is identical.  The analysis of the space bound follows
    the same lines, except that $X_i$ is the indicator
    for the event that $h_{v_i}$ is among the $k$-smallest
    elements of $\{h_{v_1},\ldots,h_{v_i}\}$, so 
    $\E(X_i)=\min\{k/i,1\}$, $\E(|S|) < kH_n$, 
    and by a Chernoff bound, $|S|=O(k\log n)$ 
    with high probability.
\end{proof}

\end{document}